\documentclass[12pt]{amsart}

\usepackage{amsfonts,color,amsmath,amssymb,fancyhdr}
\usepackage{amsthm}
\usepackage{tikz}
\usepackage{graphicx}
\usepackage{subfigure}
\usepackage{geometry}
\usepackage{cases}
\usepackage{bm}
\usepackage{booktabs}
\newtheorem{thm}{Theorem}

\newtheorem{lemma}[thm]{Lemma}
\newtheorem{corollary}[thm]{Corollary}

\numberwithin{equation}{section}
\numberwithin{thm}{section}
\newtheorem{conjecture}{Conjecture}
\theoremstyle{definition}

\DeclareMathOperator{\N}{N}
\DeclareMathOperator{\GL}{GL}
\DeclareMathOperator{\PG}{PG}
\DeclareMathOperator{\PGL}{PGL}
\DeclareMathOperator{\PGU}{PGU}
\DeclareMathOperator{\Tr}{Tr}

\begin{document}

\newcommand{\symfont}{\fam \mathfam}

\title[Linear codes associated with the Desarguesian ovoids in $Q^+(7,q)$]{Linear codes associated with the\\Desarguesian ovoids in $Q^+(7,q)$}

\author{Tao Feng}
\address{School of Mathematical Sciences, Zhejiang University, 38 Zheda Road, Hangzhou 310027, Zhejiang, P.R China}
\email{tfeng@zju.edu.cn}

\author{Michael Kiermaier}
\address{Department of Mathematics, University of Bayreuth, Universit\"{a}tsstr. 30, 95447 Bayreuth, Germany}
\email{michael.kiermaier@uni-bayreuth.de}

\author{Peixian Lin$^\ast$}
\address{School of Mathematical Sciences, Zhejiang University, 38 Zheda Road, Hangzhou 310027, Zhejiang, P.R China}
\email{lpx@zju.edu.cn}
\thanks{$^\ast$Corresponding author}

\author{Kai-Uwe Schmidt}
\address{Department of Mathematics, Paderborn University, Warburger Str.\ 100, 33098 Paderborn, Germany.}
\email{kus@math.upb.de}

\date{26 August 2022}

\begin{abstract}
The Desarguesian ovoids in the orthogonal polar space $Q^+(7,q)$ with $q$ even have first been introduced by Kantor by examining the $8$-dimensional absolutely irreducible modular representations of $\PGL(2,q^3)$. We investigate this module for all prime power values of $q$. The shortest $\PGL(2,q^3)$-orbit $O$ gives the Desarguesian ovoid in $Q^+(7,q)$ for even $q$ and it is known to give a complete partial ovoid of the symplectic polar space $W(7,q)$ for odd~$q$. We determine the hyperplane sections of $O$. As a corollary, we obtain the parameters $[q^3+1,8,q^3-q^2-q]_q$ and the weight distribution of the associated $\mathbb{F}_q$-linear code $C_O$ and the parameters $[q^3+1,q^3-7,5]_q$ of the dual code $C_O^\perp$ for $q \ge 4$. We also show that both codes $C_O$ and $C_O^\perp$ are length-optimal for all prime power values of $q$.
\end{abstract}


\maketitle

\section{Introduction}

Let $q$ be a prime power and $\mathbb{F}_q$ be the finite field with $q$ elements. An $[n,k]_q$ \emph{code} is a $k$-dimensional subspace of $\mathbb{F}_q^n$ and an $[n, k, d]_q$ \emph{code} is an $[n,k]_q$ code in which the Hamming weight of every nonzero element of the code is at least $d$. The \emph{weight distribution} of such a code is the tuple $(A_0,A_1,\dots, A_n)$, where $A_i$ is the number of codewords in the code with Hamming weight $i$. The weight distribution of a linear code contains information that is important for error detection and correction~\cite{KT2007CodesForErrorDetection}. Linear codes with few weights have important applications in cryptography~\cite{CDing2005LinearCodesandSecretsharingSchemes, CDing2006Secretsharingschemes}, authentication codes \cite{CDing2007CartesianAuthenticationcodes}, and strongly regular graphs~\cite{Kantor1986GeometryOftwoweightCodes}. Determining the weight distribution of a linear code is difficult in general and often related to interesting challenging problems in number theory \cite{IH1998CodesAndNumberTheory}.
\par
An ovoid in the orthogonal polar space $Q^{+}(2n - 1, q)$ is a subset of $q^{n - 1} + 1$ pairwise nonperpendicular points. In the case of $Q^{+}(7, q)$, Kantor~\cite{Kantor1982OvoidsAndTranslationPlanes} constructed two infinite families ovoids: the \emph{unitary} ovoid for $q \equiv 0, 2\pmod{3}$ with stabilizer $\PGU(3, q)$ and the \emph{Desarguesian} ovoid for even $q$ with stabilizer $\PGL(2, q^3)$. In both cases, the ovoids correspond to the shortest orbit of the respective group action on $\PG(7,q)$. When $q$ is odd, it is shown in \cite{PV2013DesarguesianUnitary} that the shortest $\mathrm{PGL}(2, q^3)$-orbit is a complete partial ovoid of the symplectic polar space $W(7,q)$, that is no two of its points are perpendicular and the remaining singular points are perpendicular to at least one of its points.
\par
For a set $O=\{\langle v_1\rangle,\dots,\langle v_n\rangle\}$ of $n$ points in $\PG(7,q)$, its associated linear code $C$ has generator matrix whose columns are the vectors $v_1,\dots,v_n$. Then $C$ is an $[n,k]_q$ code, where $k\le 8$. It is well known that the weight distribution of $C$ is closely related to the hyperplane sections of $O$. In \cite{Cooperstein2001HyperplaneofUnitaryOvoids}, the hyperplane sections of the unitary ovoid have been studied, giving the weight distributions of the corresponding linear codes. In this paper, we study the shortest $\PGL(2,q^3)$-orbit and the associated linear code. We determine the hyperplane sections of the orbit and thus obtain the weight distribution of the associated linear code, which is an $[q^3+1,8,q^3-q^2-q]_q$ code. Moreover we show that this code is length-optimal, namely no $[q^3,8,q^3-q^2-q]_q$ code can exist. We note that, in contrast, the codes obtained in~\cite{Cooperstein2001HyperplaneofUnitaryOvoids} have parameters $[q^3+1,8,q^3-q^2-2q]_q$ and $q$ has the restriction $q\equiv -1\pmod 6$. Also our methods are quite different compared to those in~\cite{Cooperstein2001HyperplaneofUnitaryOvoids}.
\par
This paper is organized as follows. In Section $2$, we describe the $8$-dimensional $\mathbb{F}_q$-module of $\PGL(2,q^3)$. In Section 3, we show that the group $\PGL(2,q^3)$ has four orbits when acting on $\PG(7,q)$ and determine the hyperplane sections with the shortest orbit~$O$. In Section 4, we apply our results to the linear code associated with $O$ and establish length-optimality of this code using linear programming.

\section{The $8$-dimensional $\PGL(2, q^3)$-module $V$}
\label{sec_2}

Put $V=\mathbb{F}_q\times \mathbb{F}_{q^3}\times \mathbb{F}_{q^3}\times \mathbb{F}_q$. We regard $V$ as an $8$-dimensional vector space over~$\mathbb{F}_q$ and write its elements in the form $(x, y, z, w)$, where $x, w\in \mathbb{F}_q$ and $y, z \in \mathbb{F}_{q^3}$. For a nonzero vector $v$ of $V$, we write $\langle v\rangle $ for the projective point that corresponds to $v$.
\par
Let $\Tr$ and $\N$ be the trace and norm function from $\mathbb{F}_{q^3}$ to $\mathbb{F}_q$, respectively. We define an alternating form $A:\,V\times V\rightarrow\mathbb{F}_q$ by
\[
A\left((x,y,z,w), (x',y',z',w')\right) = xw' - wx' + \Tr(zy' - yz').
\]
This form is nondegenerate and confers on $\PG(V)$ the structure of a symplectic polar space~$W(7,q)$. We also define a quadratic form $Q:V\to\mathbb{F}_q$ by
\[
Q\left((x,y,z,w)\right)=xw+\Tr(yz).
\]
This form is also nondegenerate and polarizes to $A$, that is $A(u,v)=Q(u+v)-Q(u)-Q(v)$. Hence~$Q$ confers on $\PG(V)$ the structure of an orthogonal polar space $Q^+(7,q)$.
\par
Define a subset $O$ of $\PG(V)$ by
\begin{equation}\label{eqn_Odef}
O = \{\langle(1, x, x^{q + q^2}, \N(x))\rangle:\,x \in \mathbb{F}_{q^3}\}\cup \{\langle (0,0,0,1)\rangle\}.
\end{equation}
Then $O$ is a Desarguesian ovoid of $Q^{+}(7, q)$ if $q$ is even~~\cite{Kantor1982OvoidsAndTranslationPlanes} and is a complete partial ovoid of $W(7,q)$ if $q$ is odd~\cite{PV2013DesarguesianUnitary}. Write $P(x)=\langle(1, x, x^{q + q^2}, \N(x))\rangle$ for $x\in\mathbb{F}_{q^3}$ and $P(\infty)=\langle(0,0,0,1)\rangle$. This gives a bijection between the points of $O$ and the points of the projective line $\PG(1,q^3)$.
It turns out that there is an action of $\PGL(2,q^3)$ on $V$ given by
\begin{equation}\label{eqn_g_exp}
g\,:\,\langle(x, y, z, w)\rangle\mapsto \langle(x', y', z', w')\rangle,
\end{equation}
where, writing $g = \left(\begin{smallmatrix}a & b\\c & d\end{smallmatrix}\right)$,
\begin{align*}
x'&= \N(d)x + \N(c)w + \Tr(cd^{q^2 + q}y+dc^{q^2 + q}z),\\
y'&=bd^{q^2 + q}x+ac^{q^2 + q}w + ad^{q^2 + q}y + bd^{q^2}c^qy^q + bd^qc^{q^2}y^{q^2} + bc^{q^2 + q}z + d^qac^{q^2}z^q + d^{q^2}ac^qz^{q^2},\\
z'&=db^{q^2 + q}x + ca^{q^2 + q}w + cb^{q^2 + q}y + db^{q^2}a^qy^q + db^qa^{q^2}y^{q^2} + da^{q^2 + q}z + b^qca^{q^2}z^q + b^{q^2}ca^q z^{q^2},\\
w'&=\N(b)x + \N(a)w + \Tr(ab^{q^2 + q}y +ba^{q^2 + q}z).
\end{align*}
It is routine to check that this action gives an embedding of $\PGL(2,q^3)$ into $\GL(V)$, which preserves the form $A$ and, when $q$ is even, also $Q$, namely we have $A(g(u),g(v))=A(u,v)$ for all $u,v\in V$ and, when $q$ is even, $Q(g(v))=Q(v)$ for all $v\in V$. Moreover, the action is transitive on $O$ and maps $P(x)$ to $P\left(\frac{ax+b}{cx+d}\right)$ for each $x\in \mathbb{F}_{q^3}\cup\{\infty\}$.

\section{The hyperplane sections of $O$}

Throughout this section, we use the same notation as in Section~\ref{sec_2} and denote $\PGL(2,q^3)$ by $G$. We frequently use without reference that
\[
|G|=q^3(q^3-1)(q^3+1).
\]
In this section, we determine the hyperplane sections of $O$. Each hyperplane of $V$ can be written as $v^\perp$ for some nonzero vector $v\in V$, where
\[
v^\perp=\{\langle x\rangle\in\PG(V):\,A(x,v)=0\}.
\]
Since $O$ is $G$-invariant, $g(v)^\perp\cap O$ and $v^\perp\cap O$ have the same size. Hence, it suffices to determine all $G$-orbits of $\PG(V)$ and then compute $|v^\perp\cap O|$ for one representative $v$ of each orbit.
\par
Throughout this section, we let $\alpha$ be an element of $\mathbb{F}_q$ such that $x^2-x-\alpha\in\mathbb{F}_q[x]$ is irreducible. The main result of this section is the following.
\begin{thm}\label{thm_orb_size}
There are exactly four $G$-orbits of $\PG(V)$ having the following properties (where $O$ denotes $O_1$)
\[
\begin{array}{cccc}
\toprule
\text{orbit} & \text{size} & \text{representative $v$} & |v^\perp\cap O|\\
\midrule
O_1 & q^3+1 & \langle (1,0,0,0)\rangle & 1\\[3pt]
O_2 & q(q^2 + q + 1)(q^3 + 1) & \langle(0, 0, 1, 0)\rangle & q^2+1\\[3pt]
O_3 & \frac{1}{2}q^3(q^3 + 1)(q - 1) & \langle(1,0,0,1)\rangle & q^2+q+1\\[3pt]
O_4 & \frac{1}{2}q^3(q^3 - 1)(q + 1) & \langle(1,0,\alpha,\alpha)\rangle & q^2-q+1\\
\bottomrule
\end{array}
\]
\end{thm}
\par
Here is an outline of our strategy to prove Theorem~\ref{thm_orb_size}. First it is routine to check that
\begin{align}
&|\langle(1,0,0,0)\rangle^\perp \cap O| = |\{x\in \mathbb{F}_{q^3}:\,\N(x) = 0\}| = 1 ,   \label{eqn:int_O1}\\
&|\langle(0,0,1,0)\rangle^\perp \cap O| = 1 + |\{x \in \mathbb{F}_{q^3}:\,\Tr(x) = 0\}| = q^2 + 1,   \label{eqn:int_O2}\\
&|\langle(1,0,0,1)\rangle^\perp \cap O| = |\{x \in \mathbb{F}_{q^3}:\,\N(x) = 1\}| = q^2 + q + 1.   \label{eqn:int_O3}
\end{align}
The fact that $O_1$ is a $G$-orbit has been mentioned in Section \ref{sec_2}. We first calculate the sizes of $O_2$ and $O_3$ by determining the orders of the stabilizers of chosen orbit representatives. This allows us to compute $|v^\perp\cap O|$ for $v\in O_2\cup O_3$. Then, taking $O_4$ to be the complement of $O_1\cup O_2\cup O_3$ in $\PG(V)$, we compute $|v^\perp\cap O|$ for $\langle v\rangle\in O_4$ using a counting argument. For properly chosen elements $\langle v\rangle\in O_4$, we interpret $|v^\perp\cap O|$ as the number of solutions to some equations over $\mathbb{F}_{q^3}$. This latter information is then used to determine the number of solutions over~$\mathbb{F}_{q^3}$ of two equations, which is then used to show that $O_4$ is indeed a $G$-orbit.
\par
The remaining part of this section is devoted to the proof of Theorem~\ref{thm_orb_size}.
\begin{lemma}\label{lem_O2_stab}
The orbit $O_2$ that contains $\langle(0,0,1,0)\rangle$ has size $q(q^2 + q + 1)(q^3 + 1)$.
\begin{proof}
Suppose that $g\in \GL(2,q^3)$ stabilizes  $\langle(0,0,1,0)\rangle$. Writing $g = \left(\begin{smallmatrix}a & b\\c & d\end{smallmatrix}\right)$, in view of~\eqref{eqn_g_exp} this is the case if and only if there exists $\lambda \in \mathbb{F}^*_q$ such that
\begin{align}
0 &= \Tr(dc^{q^2+q}),\label{equ:3.2.1}\\
0 &= bc^{q^2+q} + d^qac^{q^2} + d^{q^2}ac^q, \label{equ:3.2.2}\\
\lambda &= da^{q^2+q} + b^qa^{q^2}c + b^{q^2}a^qc, \label{equ:3.2.3}\\
0 &= \Tr(ba^{q^2+q}).\notag
\end{align}
\par
We now distinguish three cases.
\par
{\itshape Case 1: $c = 0$.} Then the above conditions reduce to $\mathrm{Tr}(ba^{q^2+q}) = 0$ and $da^{q^2+q} = \lambda$. Since $\det(g)\ne 0$ we have $ad\ne 0$. For each pair $(a,\lambda)\in\mathbb{F}_{q^3}^*\times\mathbb{F}_q^*$, there is exactly one $d\in\mathbb{F}_{q^3}$ and $q^2$ elements $b\in\mathbb{F}_{q^3}$ such that the two conditions are satisfied. Therefore, there are $q^2(q^3-1)(q-1)$ such $g\in\textup{GL}(2,q^3)$ and correspondingly $q^2(q-1)$ elements in~$G$ that stabilize $\langle(0,0,1,0)\rangle$.
\par
{\itshape Case 2: $cd\ne 0$.} Since $\GL(2,q^3)$ is acting on projective points, we may assume that $d = 1$. Write $u=1/c$. Then \eqref{equ:3.2.1} and \eqref{equ:3.2.2} reduce to $\Tr(u) = 0$ and $b=-au^q -au^{q^2}$. We deduce that
\begin{align*}
a^{q^2 + q} + b^qa^{q^2}c+ b^{q^2}a^qc
=&a^{q^2 + q} + (-a^qu^{q^2}-a^qu)a^{q^2}u^{-1}+(-a^{q^2}u-a^{q^2}u^q)a^qu^{-1}\\
=&a^{q^2+q}u^{-1}(u-u^{q^2}-u-u-u^{q})\\
=&-a^{q^2+q}u^{-1}\textup{Tr}(u)=0,
\end{align*}
which contradicts~\eqref{equ:3.2.3}.
\par
{\itshape Case 3: $c \ne 0$ and $d = 0$.} In this case we deduce that $b = 0$ from \eqref{equ:3.2.2}. It follows that the left hand side of ~\eqref{equ:3.2.3} equals zero, a contradiction.
\par
In summary the stabilizer of $\langle(0,0,1,0)\rangle$ in $G$ has order $q^2(q - 1)$. The claim now follows from the orbit stabilizer theorem.
\end{proof}
\end{lemma}
\par
\begin{lemma}\label{lem_O3_stab}
The orbit $O_3$ that contains $\langle(1, 0, 0, 1)\rangle$ has size $\frac{1}{2}q^3(q^3 + 1)(q - 1)$.
\end{lemma}
\begin{proof}
Suppose that $g\in \GL(2,q^3)$ stabilizes  $\langle(0,0,1,0)\rangle$. Writing $g = \left(\begin{smallmatrix}a & b\\c & d\end{smallmatrix}\right)$, we find from~\eqref{eqn_g_exp} that this is the case if and only if there exists an element $\lambda \in \mathbb{F}^*_q$ such that
\begin{align}
\lambda &= \N(c) + \N(d), \label{equ:3.10}\\
0 &= ac^{q^2 + q}+ bd^{q^2 + q}, \label{equ:3.11}\\
0 &= ca^{q^2 + q} + db^{q^2 + q}, \label{equ:3.12}\\
\lambda &= \N(a) + \N(b). \label{equ:3.13}
\end{align}
We distinguish three cases.
\par
{\itshape Case 1: $c = 0$.} In this case we get $b = 0$ from \eqref{equ:3.11} and the fact $ad-bc\ne 0$. Then \eqref{equ:3.12} holds trivially, and \eqref{equ:3.10} and \eqref{equ:3.13} reduce to  $\N(d) = \N(a) =\lambda$. There are $(q^3-1)(q^2+q+1)$ possible choices for $g$ and hence $q^2+q+1$ elements of $G$ that stabilize $\langle(1,0,0,1)\rangle$.
\par
{\itshape Case 2: $d = 0$.} This case is similar to the first case. We have $a=0$ and find again that there are $q^2+q+1$ elements of $G$ that stabilize $\langle(1,0,0,1)\rangle$.
\par
{\itshape Case 3: $cd \neq 0$.} Then \eqref{equ:3.11} and $ad-bc\ne0$ imply that $ab\ne0$. Write $u=d/c$ and $m=a/b$. Then \eqref{equ:3.11} and \eqref{equ:3.12} reduce to $m =-u^{q^2 + q}$ and $m^{q^2 + q}=-u $. It follows that
\[
-m=u^{q^2 + q}=m^{(q^2 + q)(q^2+q)}=m\cdot\textup{N}(m)
\]
and therefore $\N(m)=-1$. We thus have
\[
\N(a)+\N(b)=\N(b)(\N(m)+1)=0,
\]
which contradicts~\eqref{equ:3.13}.
\par
In summary the stabilizer of $\langle (1,0,0,1)\rangle$ in $G$ has order $2(q^2 + q + 1)$ and the claim now follows from the orbit stabiliser theorem
\end{proof}
\par
In what follows we put
\[
O_4=\PG(V)\backslash (O_1\cup O_2 \cup O_3).
\]
We shall now show in a series of lemmas that $O_4$ is indeed a $G$-orbit containing $\langle(1,0,\alpha,\alpha)\rangle$. We begin with determining the hyperplane sections of $O$ corresponding to points of $O_4$.
\begin{lemma}\label{lemma:O4 cap O}
For each $v \in O_4$, we have $|v^\perp \cap O| = q^2 - q + 1$.
\end{lemma}
\begin{proof}
We double count the triples $(\langle v\rangle, \langle x\rangle, \langle y\rangle)$, where $\langle v\rangle \in \PG(V)$ and $\langle x\rangle,\langle y\rangle$ are distinct points of $v^\perp\cap O$. On one hand, the number of such triples is
\[
\sum_{v \in \PG(V)}|v^\perp \cap O|(|v^\perp \cap O|-1).
\]
On the other hand, for distinct points $\langle x\rangle,\langle y\rangle$ in $O$, the size of $\langle x,y\rangle^\perp$ equals $(q^6-1)/(q-1)$ and thus we thus have
\begin{align}\label{equ:18}
\sum_{v \in \PG(V)}|v^\perp \cap O| (|v^\perp \cap O|-1)= q^3(q^3 + 1)\frac{q^6 - 1}{q - 1}.
\end{align}
Similarly, by double counting the pairs $(\langle v\rangle, \langle x\rangle)$, where $\langle v\rangle \in \PG(V)$ and $\langle x\rangle\in v^\perp\cap O$, we obtain
\begin{align}
\sum_{v \in \textup{PG}(V)}|v^\perp \cap O| &= (q^3 + 1)\frac{q^7 - 1}{q - 1} \label{equ:19}.
\end{align}
By combining~\eqref{equ:18} and~\eqref{equ:19} with~\eqref{eqn:int_O1},~\eqref{eqn:int_O2},~\eqref{eqn:int_O3} and Lemmas~\ref{lem_O2_stab} and~\ref{lem_O3_stab}, we find that
\begin{align*}
&\sum_{v \in O_4}|v^\perp \cap O| = \frac{1}{2}q^3(q^3 - 1)(q^3 + 1), \\
&\sum_{v \in O_4}|v^\perp \cap O|^2 = \frac{1}{2}q^3(q^3 - 1)(q^3 + 1)(q^2 - q + 1).
\end{align*}
Since $O_4$ has size $\frac{1}{2}q^3(q^3 - 1)(q+ 1)$, we deduce that
\[
\sum_{v \in O_4}\left(|v^\perp \cap O|-(q^2-q+1)\right)^2=0,
\]
which completes the proof.
\end{proof}
\par
The next step will be to determine the number of solutions of certain equations, for which we need the following lemma.
\begin{lemma}\label{lemma:two points}
We have $\langle(0, 1, -1, 1 + \alpha)\rangle, \langle(1, \alpha, 0, \alpha^2)\rangle\in O_4$.
\end{lemma}
\begin{proof}
By~\eqref{eqn_Odef} it is clear that the two points are not in $O_1$. It thus suffices to show that the two points are not in $O_2\cup O_3$. We show that this holds for $P=\langle(0, 1, -1, 1 + \alpha)\rangle$. The other point can be treated similarly.
\par
First suppose for a contradiction that $P \in O_2$. Then there exists $g=\left(\begin{smallmatrix}a & b\\c & d\end{smallmatrix}\right)$ in $\GL(2,q^3)$ that maps $\langle(0, 0, 1, 0)\rangle$ to $\langle (0, 1, -1, 1 + \alpha) \rangle$. By \eqref{eqn_g_exp}, there exists $\lambda \in \mathbb{F}^*_q$ such that
\begin{align}
0 &= \Tr(dc^{q^2 + q}),\label{equ:3.15}\\
\lambda &= bc^{q^2 + q} + d^qac^{q^2} + d^{q^2}ac^q,\label{equ:3.16}\\
-\lambda &= da^{q^2 + q} + b^qca^{q^2} + b^{q^2}ca^q, \label{equ:3.17}\\
\lambda(1 + \alpha) &= \Tr(ba^{q^2 + q}).\label{eqn_3_18}
\end{align}
From~\eqref{equ:3.16} and~\eqref{equ:3.17} we deduce that $ac\neq 0$. Since $\GL(2,q^3)$ is acting on projective points, we may assume that $c=1$. Then \eqref{equ:3.15} takes the form $\mathrm{Tr}(d) = 0$, hence $d+d^q+d^{q^2}=0$. Then~\eqref{equ:3.16} reduces to $b-ad=\lambda$.
\par
Substitute $b=\lambda+ad$ into the left hand side of~\eqref{equ:3.17} to obtain
\begin{align*}
da^{q^2 + q} + b^qa^{q^2} + b^{q^2}a^q
&=da^{q^2 + q} + (\lambda+a^qd^q)a^{q^2} + (\lambda+a^{q^2}d^{q^2})a^q\\
&=\Tr(d)a^{q^2+q}+\lambda(a^q+a^{q^2})\\
&=\lambda(a^q+a^{q^2}).
\end{align*}
Hence~\eqref{equ:3.17} reduces to $\lambda(a^q+a^{q^2})=-\lambda$, and so $a=1+\Tr(a)$. Therefore $a\in\mathbb{F}_q$ and $a=1+3a$. This implies that $q$ is odd and $a=-1/2$. Using $b=\lambda+ad$ and $\Tr(d)=0$, \eqref{eqn_3_18} then reduces to
\[
\tfrac{3}{4}\lambda=\lambda(1+\alpha).
\]
Hence $\alpha=-1/4$. But then $x^2-x-\alpha=(x-1/2)^2$, contradicting the assumed irreducibility of this polynomial. This proves that $P\not\in O_2$.
\par
Now suppose for a contradiction that $P\in O_3$. Then there exists $g=\left(\begin{smallmatrix}a & b\\c & d\end{smallmatrix}\right)$ in $\GL(2,q^3)$ that maps $\langle(1, 0, 0, 1)\rangle$ to $\langle (0, 1, -1, 1 + \alpha) \rangle$. By \eqref{eqn_g_exp}, this implies that there exists $\lambda\in\mathbb{F}_q$ such that
\begin{align}
0 &= \N(c) + \N(d),\label{equ:3x} \\
\lambda &= ac^{q^2 + q} + bd^{q^2 + q},\label{equ:3.18}\\
-\lambda &= ca^{q^2 + q} + db^{q^2 + q},\label{equ:3.19}\\
\lambda(1 + \alpha) &= \N(a) + \N(b). \label{eqn_3_20}
\end{align}
From~\eqref{equ:3x} and $ad-bc\ne 0$ we deduce that $cd\ne 0$. Again we may assume that $c=1$. Since $\lambda^q = \lambda$, the same holds for the left hand sides of \eqref{equ:3.18} and \eqref{equ:3.19}.
\begin{align*}
a+ bd^{q^2 + q} &= a^q+b^qd^{q^2 + 1},\\
a^{q^2 + q} + db^{q^2 + q} &= a^{q^2 + 1} + d^qb^{q^2 + 1}.
\end{align*}
After rearranging terms, we have
\begin{align}
a-a^q&=d^{q^2}(db^q-bd^q),   \label{equ:adb}\\
a^{q^2}(a-a^q)&=b^{q^2}(db^q-bd^q).   \notag
\end{align}
from which we find that
\[
(b- ad)^{q^2}(a - a^q)(db^q - bd^q) = 0.
\]
From $bc - ad \ne 0$ and $c=1$ and~\eqref{equ:adb}, we then conclude $a- a^q = 0$ and so $a\in\mathbb{F}_q$.
\par
From $c=1$ and~\eqref{equ:3.15} we have $\N(d)=-1$. From~\eqref{equ:3.18} we have $\lambda=a-b/d$ and thus $b/d=a-\lambda$ is in $\mathbb{F}_q$. Hence
\[
\N(b)=\N(d)\N(b/d)=-(b/d)^3.
\]
By adding up both sides of \eqref{equ:3.18} and \eqref{equ:3.19}, we obtain
\[
a+a^2-(b/d)-(b/d)^2=0,
\]
or equivalently $(a-b/d)(a+b/d+1)=0$. Since $a-b/d=\lambda\ne 0$, we have $a+b/d+1=0$. From \eqref{eqn_3_20} we then deduce that
\begin{align*}
\alpha&=\frac{\N(a)-\N(b/d)}{a-b/d}-1\\
&=\frac{a^3-(b/d)^3}{a-b/d}-1\\
&=a^2+ab/d+(b/d)^2-1\\
&=a(a+1),
\end{align*}
using $a=-b/d-1$. But then $x^2-x-\alpha=(x-a)(x+a+1)$, which again contradicts the irreducibility of this polynomial. This proves that $P\not\in O_3$.
\end{proof}
\par
\begin{corollary}\label{coro:point cap size}
We have
\begin{align*}
|\{x\in\mathbb{F}_{q^3}:\,\Tr(x) + \Tr(x^{q^2 + q})  + 1 + \alpha=0\}|&=q^2-q,\\
|\{x\in\mathbb{F}_{q^3}:\,\N(x) - \Tr(x^{q^2 + q})\alpha - \alpha^2 = 0\}|&=q^2-q+1.
\end{align*}
\end{corollary}
\begin{proof}
For each point $P= \langle(v_1, v_2, v_3, v_4)\rangle$ in $\PG(V)$, the set $P^\perp \cap (O \backslash \{\langle(0,0,0,1)\rangle\}) $ equals
\[
\{\langle(1, x, x^{q + q^2}, \textup{N}(x))\rangle:v_4 - v_1\textup{N}(x)+ \textup{Tr}(v_2x^{q^2 + q} - v_3x) = 0, x\in \mathbb{F}_{q^3}\}.
\]
By taking $P$ as $\langle(0, 1, -1, 1 + \alpha)\rangle$ and $\langle(1, \alpha, 0, \alpha^2)\rangle$, the result follows from Lemmas~\ref{lemma:O4 cap O} and~\ref{lemma:two points} and the fact that $P^\perp$ contains $\langle(0,0,0,1)\rangle$ if and only if $v_1=0$.
\end{proof}
\par
Now we show that $O_4$ is indeed a $G$-orbit.
\begin{lemma}\label{lem_O4_sz}
The set $O_4$ is a $G$-orbit of size $\frac{1}{2}(q^3(q^3 - 1)(q + 1))$ and contains $\langle(1, 0, \alpha, \alpha)\rangle$.
\end{lemma}
\begin{proof}
Suppose that $g = \left(\begin{smallmatrix}a & b\\c & d\end{smallmatrix}\right)$ in $\GL(2,q^3)$ stabilizes  $P=\langle(1,0,\alpha,\alpha)\rangle$. By \eqref{eqn_g_exp}, this is the case if and only if there exists $\lambda \in \mathbb{F}^*_q$ such that
\begin{align}
\lambda &= \N(d) + \N(c)\alpha + \Tr(dc^{q^2 + q})\alpha, \label{equ:3.20}\\
0 &= ac^{q^2 + q}\alpha + bd^{q^2 + q} + (bc^{q^2 + q} + d^qac^{q^2} + d^{q^2}ac^q)\alpha \label{equ:3.21},\\
\lambda\alpha &= ca^{q^2 + q}\alpha + db^{q^2 + q} + (da^{q^2 + q} + b^qca^{q^2} + b^{q^2}ca^q)\alpha, \label{equ:3.22}\\
\lambda\alpha &= \N(b) + \N(a)\alpha + \Tr(ba^{q^2 + q})\alpha \label{equ:3.23}.
\end{align}
First observe that $\alpha\ne-1/4$, since otherwise $x^2-x+1/4=(x-1/2)^2$ is not irreducible. Next we distinguish four cases.
\par
{\itshape Case 1: $c = 0$.} In this case we have $b = 0$ from \eqref{equ:3.21} and $ad-bc\ne 0$. Then~\eqref{equ:3.22} and~\eqref{equ:3.23} reduce to $(d/a)\N(a)=\N(a)$ and so $a=d$. Hence $g$ is a scalar matrix and corresponds to the identity in $G$.
\par
{\itshape Case 2: $d = 0$.} Then $c\ne 0$ and we may assume that $c=1$. Then~\eqref{equ:3.21} implies $b = -a$. From~\eqref{equ:3.20} we have $\lambda=\alpha$ and so $\alpha\in\mathbb{F}_q$. From~\eqref{equ:3.22} we find that $\N(a)=-a\alpha$, which implies $a\in\mathbb{F}_q$. From~\eqref{equ:3.23} we then find that $a=\alpha/(2\alpha+1)$ and $\alpha=-1$ upon taking norms. To sum up, we obtain a unique element in $G$ for $\alpha=-1$ and a contradiction otherwise.
\par
{\itshape Case 3: $cd \ne 0$ and $d^{q^2 + q} + \alpha =0$.}  Again we may assume that $c = 1$. Since $\N(d)=-\alpha d$, we deduce that $d \in \mathbb{F}_q^*$ and so $d^2 = -\alpha$. As remarked earlier this implies $d \neq \pm 1/2$. Now~\eqref{equ:3.21} reduces to $a\alpha(1 + 2d) = 0$, so $a = 0$. The other three equations then reduce to $d = b$ and $(2d - 1)(d + 1) = 0$. It follows that $d = -1$ and $\alpha=-1$. Hence we again obtain a unique element in $G$ for $\alpha=-1$ and a contradiction otherwise.
\par
{\itshape Case 4: $cd(d^{q^2 + q} + \alpha) \ne 0$.} As usual, we may assume that $c=1$. Now the conditions are equivalent to
\begin{align}
\N(d) + \alpha + \Tr(d)\alpha&\ne 0,   \label{eqn:det_condition}\\
a\alpha + bd^{q^2 + q} + (b + d^qa + d^{q^2}a)\alpha &= 0 \label{equ:3.25},\\
a^{q^2 + q}\alpha + db^{q^2 + q} + (da^{q^2 + q} + b^qa^{q^2} + b^{q^2}a^q)\alpha&= \N(b)+\N(a)\alpha +\Tr(ba^{q^2 + q})\alpha, \label{equ:3.26}\\
a^{q^2 + q}\alpha + db^{q^2 + q} + (da^{q^2 + q} + b^qa^{q^2} + b^{q^2}a^q)\alpha&=\N(d)\alpha + \alpha^2 + \Tr(d)\alpha^2. \label{equ:3.27}
\end{align}
Put
\[
u=\frac{1 + d^q + d^{q^2}}{d^{q^2 + q} + \alpha}.
\]
Then~\eqref{equ:3.25} implies that $b = -ua\alpha$, hence $ab\ne 0$. Substitute into~\eqref{equ:3.26} and divide both sides by $a^{q^2+q}\alpha$ to obtain
\[
1 + d + (du^{q^2 + q} - u^q - u^{q^2})\alpha=(1  - \N(u)\alpha^2 -\Tr(u)\alpha)a.
\]
We plug in the expression of $u$ to see that the left hand side equals
\begin{equation}\label{eqn_case4_LHS}
\frac{(d + d^2 - \alpha)(\N(d) + \alpha + \Tr(d)\alpha)}{(d^{q^2 + 1} + \alpha)(d^{q + 1} + \alpha)}
\end{equation}
and the right hand side equals
\begin{align*}
\frac{(\mathrm{N}(d) - 2\alpha\mathrm{Tr}(d^{q^2 + q}) - \alpha\mathrm{Tr}(d) - 2\alpha^2 - \alpha)(\mathrm{N}(d) + \alpha + \mathrm{Tr}(d)\alpha)}{\textup{N}(d^{q^2 + q} + \alpha)}\,a.
\end{align*}
Hence we have
\[
a=\frac{(d^{q^2 + q} + \alpha)(d^2 + d - \alpha)}{\N(d) - 2\alpha\Tr(d^{q^2 + q}) - \alpha\Tr(d) - 2\alpha^2 - \alpha}.
\]
We now have expressed $a$ and $b$ in terms of $d$. From ~\eqref{eqn_case4_LHS} we find that \eqref{equ:3.27} reduces to
\[
\alpha(\textup{N}(d) + \alpha + \textup{Tr}(d)\alpha)\left(\frac{\textup{N}(d + d^2 -\alpha)}{(\mathrm{N}(d) - 2\alpha\mathrm{Tr}(d^{q^2 + q}) - \alpha\mathrm{Tr}(d) - 2\alpha^2 - \alpha)^2} - 1\right) = 0.
\]
Next we deduce that
\[
\det(g) = a(d + u\alpha) = \frac{(d^2 + d  - \alpha)(\N(d) + \alpha + \Tr(d)\alpha)}{\N(d) - 2\alpha\Tr(d^{q^2 + q}) - \alpha\Tr(d) - 2\alpha^2 - \alpha}.
\]
Hence $\det(g) \ne 0$ is equivalent to~\eqref{eqn:det_condition}. Therefore the conditions on $g$ reduce to the following: $d(d^{q^2 + q} + \alpha) \neq 0$, $\N(d) + \alpha + \Tr(d)\alpha  \ne 0$, and
\begin{align}
\N(d + d^2 - \alpha) - (\N(d) - 2\alpha\Tr(d^{q^2 + q}) - \alpha\Tr(d) - 2\alpha^2 - \alpha)^2 = 0.\label{equ:3.28}
\end{align}
Put
\begin{align*}
x&=\Tr(d) + \Tr(d^{q^2 + q}) + 1 + \alpha,\\
y&=\N(d) - \Tr(d^{q^2 + q})\alpha - \alpha^2,
\end{align*}
so that
\begin{align*}
y - \alpha x &= \N(d) - 2\alpha\Tr(d^{q^2 + q}) - \alpha\Tr(d) - 2\alpha^2 - \alpha,\\
y+\alpha x &= \N(d) + \Tr(d)\alpha + \alpha.
\end{align*}
It is tedious but not difficult to check that $\textup{N}(d + d^2 - \alpha)= (y + \alpha x)^2+xy$, so that~\eqref{equ:3.28} reduces to $(4\alpha+1)xy=0$, which in turn reduces to $xy=0$ using $\alpha\ne -1/4$. Therefore the conditions on $g$ reduce to
\begin{equation}
xy=0,\qquad d(d^{q^2 + q} + \alpha) \ne 0, \qquad \N(d) + \Tr(d)\alpha + \alpha\ne 0.   \label{eqn:conditions_case4}
\end{equation}
First suppose that $x=0$. By Corollary \ref{coro:point cap size} there $q^2-q$ choices for $d$ such that $x=0$. We now show that the second and third condition in~\eqref{eqn:conditions_case4} are satisfied for these choices of $d$ unless $\alpha=-1$ and $d=0$ or $-1$.
\par
It is routine to check that $d^{q^2+q}+\alpha=-(1+d+d^q)(1+d)$ and $\N(d) + \Tr(d)\alpha + \alpha=-\N(1+d+d^q)$. The condition $1+d+d^q=0$ is equivalent to $d=1+\Tr(d)$ and thus to $2d=-1$. Hence the second and third conditions in~\eqref{eqn:conditions_case4} are equivalent to $d(1+d)(2d+1)\ne0$. If $2d+1=0$, then $q$ is odd and $d=-1/2$ and $x=0$ forces $\alpha=-1/4$, a contradiction. If $d=0$ or $d=-1$, then $\alpha=-1$. Hence, for $x=0$, we obtain $q^2-q$ suitable elements of $G$, unless $\alpha=-1$, in which case we obtain $q^2-q-2$ such elements.
\par
Note that the third condition in~\eqref{eqn:conditions_case4} is equivalent to $y+\alpha x\ne 0$. Thus $x=0$ forces $y\ne 0$ (this also holds for $\alpha=-1$ and $d=0$ or $-1$). Suppose now that $y=0$. By contraposition this forces $x\ne 0$. By Corollary \ref{coro:point cap size} there $q^2-q+1$ choices for $d$ such that $y=0$. We now show that the second and third condition in~\eqref{eqn:conditions_case4} are satisfied for these choices of $d$.
\par
It is clear that $d\ne 0$, and so  $d^{q^2 + q} + \alpha\ne 0$ by the fact $\N(d)= (\Tr(d^{q^2 + q})+\alpha)\alpha$. Hence the second condition in~\eqref{eqn:conditions_case4} holds. Since $x\ne 0$, we have $y+\alpha x\ne 0$ and so the third condition in~\eqref{eqn:conditions_case4} is also satisfied. Hence, for $y=0$, we obtain $q^2-q+1$ suitable elements of $G$
\par
In summary the four cases give exactly $2(q^2 - q + 1)$ elements in $G$ that stabilize $\langle v\rangle$, from which the claim follows.
\end{proof}
\par
Theorem~\ref{thm_orb_size} now follows by combining Lemmas~\ref{lem_O2_stab},~\ref{lem_O3_stab},~\ref{lemma:O4 cap O}, and~\ref{lem_O4_sz}.

\section{The linear code of the set $O$}

It is well known that up to linear equivalence, spanning multisets of points in $\PG(\mathbb{F}_q^k) \cong \PG(k-1,q)$ correspond to linear codes over $\mathbb{F}_q$ of dimension $k$ in which no coordinate is identically zero. Thereby, the hyperplanes correspond to the projective equivalence classes of the nonzero codewords and the hyperplane sections correspond to the Hamming weights. An early publication of this observation is \cite[Thm.~1.11]{Burton1964Thesis} and a good overview is given in \cite{DodunekovSimonis1998CodesAndProjectiveMultisets}.
\par
To make this correspondence precise, let $\mathcal{P}$ be a multiset of $n$ points $\langle v_1\rangle,\ldots,\langle v_n\rangle$ spanning $\PG(\mathbb{F}_q^k)$. With $\mathcal{P}$ we associate the $[n,k]_q$ code $C$ generated by the matrix $G$ with the columns $v_1,\ldots,v_n$. While this code depends on the chosen point representatives and the order of the points, it is unique up to linear equivalence. Conversely every equivalence class of $[n,k]_q$ codes in which no coordinate is identically zero corresponds to a multiset of $n$ points spanning $\PG(\mathbb{F}_q^k)$. A hyperplane $H = v^\perp$ of $\PG(\mathbb{F}_q^k)$ corresponds to the set of $q-1$ nonzero scalar multiples of $v^\top G$, which are all codewords of $C$. The (common) Hamming weight of these $q-1$ codewords is given by the number of points in $\mathcal{P}$ that are not contained in $H$.
\par
Now based on the notation of Section~\ref{sec_2}, we are going to describe and investigate the code $C_O$ associated with the point set $\mathcal{P} = O$ in $\PG(\mathbb{F}_q^8)$. Explicitly, the code $C_O$ is generated by
\[
G_O = \begin{bmatrix}
g(x_1) & g(x_2) & \cdots & g(x_{q^3}) & e_8
\end{bmatrix},
\]
where $e_8 = (0, 0, 0, 0, 0, 0, 0, 1)^\top$ and
\[
g(x)=(1,\Tr(x),\Tr(\theta x),\Tr(\theta^2 x),\Tr(x^{q + q^2}),\Tr(\theta x^{q + q^2}),\Tr(\theta^2 x^{q + q^2}),\N(x))^\top
\]
with $\theta$ being a primitive element of $\mathbb{F}_{q^3}$. Then $C$ has length $q^3+1$ and dimension $8$, using the fact that $O$ spans $\PG(\mathbb{F}_q^8)$. The weight distribution (and therefore the minimum distance) of $C_O$ follows from the above discussion and the hyperplane sections given in Theorem~\ref{thm_orb_size}.
\begin{thm}
\label{thm:CO}
The code $C_O$ has parameters $[q^3 + 1, 8, q^3 - q^2 - q]_q$ and the following weight distribution
\[
\begin{array}{cc}
\toprule
weight & multiplicity\\
\midrule
0              & 1 \\ [3pt]
q(q^2 - q - 1) & \frac{1}{2}q^3(q^3 + 1)(q -1)^2\\ [3pt]
q^2(q - 1)     & q(q^6 - 1)\\[3pt]
q(q^2 - q + 1) & \frac{1}{2}q^3(q^3 - 1)(q^2 - 1)\\[3pt]
q^3            & (q^3+1)(q - 1)\\
\bottomrule
\end{array}
\]
\end{thm}
\par
\begin{table}[htbp]
	\centering
	\caption{THE CODES $C_O$ FOR SMALL $q$}
    \label{table_code_small_q}
	\begin{tabular}{lll}
		\toprule
		$q$ & parameters & weight distribution \\
		\midrule
		$2$ & $[9,8,2]_2$ & $(0^1 2^{36} 4^{126} 6^{84} 8^9)$ \\ [3pt]
		$3$ & $[28,8,15]_3$ & $(0^1 15^{1512} 18^{2184} 21^{2808} 27^{56})$ \\ [3pt]
		$4$ & $[65,8,44]_4$ &  $(0^1 44^{18720} 48^{16380} 52^{30240} 64^{195})$ \\ [3pt]
		$5$ & $[126,8,95]_5$ & $(0^1  95^{126000} 100^{78120} 105^{186000} 125^{504})$ \\ [3pt]
		$7$ & $[344,8,287]_7$ & $(0^1 287^{2123856} 294^{823536} 301^{2815344} 343^{2064})$ \\ [3pt]
		$8$ & $[513,8,440]_8$ & $(0^1 440^{6435072} 448^{2097144} 456^{8241408} 512^{3591})$ \\ [3pt]
		$9$ & $[730,8,639]_9$ & $(0^1 639^{17029440} 648^{4782960} 657^{21228480} 729^{5840})$ \\
		\bottomrule
	\end{tabular}
\end{table}

For small $q$, the parameters and the weight distribution of the code $C_O$ are given explicitly in Table~\ref{table_code_small_q}.

Next we discuss optimality properties of $C_O$ with respect to the theory in~\cite{DodunekovSimonis2000}. For a linear code $C$ with parameters $[n,k,d]_q$, the following notions of optimality are defined in~\cite{DodunekovSimonis2000}:
\begin{itemize}
\item $C$ is \emph{length-optimal} (\emph{$n$-optimal}) if no linear $[n-1,k,d]_q$ code exists;
\item $C$ is \emph{dimension-optimal} (\emph{$k$-optimal}) if no linear $[n,k+1,d]_q$ code exists;
\item $C$ is \emph{distance-optimal} (\emph{$d$-optimal}) if no linear $[n,k,d+1]_q$ code exists;
\item $C$ is \emph{shortening-optimal} (\emph{$S$-optimal}) if no linear $[n+1,k+1,d]_q$ code exists;
\item $C$ is \emph{puncturing-optimal} (\emph{$P$-optimal}) if no linear $[n+1,k,d+1]_q$ code exists.
\end{itemize}
There are the following implications:
\begin{align*}
\text{$C$ is $n$-optimal} & \implies \text{$C$ is $k$-optimal and $d$-optimal};\\
\text{$C$ is $P$-optimal} & \implies \text{$C$ is $d$-optimal};\\
\text{$C$ is $S$-optimal} & \implies \text{$C$ is $k$-optimal}.
\end{align*}
Suitable examples show that there are no other implications among the five optimality notions. If a code is optimal with respect to all five of the above optimality concepts, then it is called \emph{strongly optimal}. By the stated implications, strong optimality already follows from $n$-, $S$- and $P$-optimality.
\par
For $q = 2$, $C_O$ is the binary parity check code with the parameters $[9,8,2]_2$, which is an MDS-code and in particular $n$-optimal.
Moreover, it is $P$-optimal, but not $S$-optimal as there exists the $[10,9,2]_2$ binary parity check code.
\par
We now consider the case $q\ge 3$. The $i$th \emph{Krawtchouk polynomial} (with parameters $q$ and $n$) is the polynomial
\[
K_i = \sum_{j=0}^i (-1)^j(q - 1)^{(i - j)}\binom{x}{j} \binom{n - x}{i - j}
\]
of degree $i$ in $\mathbb{R}[x]$. For each polynomial $f\in\mathbb{R}[x]$ of degree at most $n$ there are unique $f_0,f_1,\dots,f_n\in\mathbb{R}$ such that
\[
f=\sum_{i=0}^{n}f_i K_i,
\]
which is known as the \emph{Krawtchouk expansion} of $f$. Next we quote the well known linear programming bound (see~\cite{MacWilliams1977Codes}, for example).
\begin{thm}
\label{thm:lp_bound}
Let $C$ be a code of length $n$ and minimum distance $d$ over an alphabet of size $q$. Let $f\in\mathbb{R}[x]$ be a polynomial with Krawtchouk expansion $f=\sum_{i=0}^nf_i K_i$ such that $f_i\ge0$ for all $i\in\{0,1,\ldots,n\}$ and $f(i)\le 0$ for all $i\in\{d,d+1,\ldots,n\}$. Then we have
\[
|C| \le f(0)/f_0.
\]
\end{thm}
\par
We now apply Theorem~\ref{thm:lp_bound} to show that $C_O$ is $n$-optimal for all $q$.
\begin{thm} \label{thm:4.2}
Let $C$ be a code of length $q^3$ and minimum distance at least $q^3 - q^2 - q$ over a $q$-ary alphabet for $q \geq 3$. Then
$$|C| \leq \frac{1}{2} \frac{q^5(q^2 - q - 1)(q^3 - q^2 + q - 2)(q^2 + 1)}{q^4 - 2q^3 - q^2 + 3}.$$
In particular $|C|<q^8$ and so no $[q^3, 8, q^3 - q^2 - q]_{q}$ code exists.
\end{thm}
\begin{proof}
Write
\[
z_1 = q^3 - q^2 - q, \quad
z_2 = q^3 - q^2 + q - 2, \quad
z_3 = q^3 - q^2 + q - 1,\quad\text{and}\quad
n = q^3.
\]
We have $0 < z_1 < z_2 < z_3 < n$, so the polynomial $f(x) = (x-z_1)(x-z_2)(x-z_3)(x-n)$ has no repeated roots. Note that $C$ has minimum distance at least $z_1$ and length $n$. Since $f(0)>0$, we therefore have $f(i) \le 0$ for all $i\in\{d,d+1,\dots,n\}$. The Krawtchouk expansion of $f$ is
\[
f = f_0K_0 + f_1K_1 + f_2K_2 + f_3K_3 + f_4K_4,
\]
where
\begin{align*}
f_0 &=\frac{2}{q}(q - 1)(q^4 - 2q^3 - q^2 + 3), \\
f_1 &= \frac{2}{q^4}(q - 1)(q^6 + q^5 - 10q^3 + 3q + 12),\\
f_2 &= \frac{2}{q^4}(q^5 + 5q^4 - 9q^3 - 6q^2 - 18q + 36),\\
f_3 &= \frac{6}{q^4}(q^3 + q^2 + 3q - 12),\\
f_4 &= \frac{24}{q^4}.
\end{align*}
Since $q\ge 3$, all these coefficients are nonnegative. Then Theorem~\ref{thm:lp_bound} gives the desired bound for $|C|$.
\end{proof}
\par
Theorem~\ref{thm:4.2} gives the following corollary.
\begin{corollary}
\label{cor:n-optimality}
The code $C_O$ is $n$-optimal for all $q$.
\end{corollary}
\par
For $q\in\{3,4,5\}$, we can say a bit more.
\begin{thm}
The code $C_O$ is strongly optimal for $q\in\{3,4\}$ and $S$-optimal for $q = 5$.
\end{thm}
\begin{proof}
The nonexistence of the corresponding codes can be looked up at the online tables~\cite{codetables}. We illustrate the typical reasoning behind the entries in the case $q = 3$.
\par
Assume that $C_O$ is not $P$-optimal. Then there exists a $[29,8,16]_3$ code $C$. Its residual (obtained by shortening~$C$ in the support of a codeword of weight $16$) is a $[13,7,\ge 6]_3$ code~\cite[Thm.~2.7.1]{HuffmanPless}. Puncturing in a single position yields a $[12,7,\ge 5]_3$ code. However according to the sphere packing bound, a ternary code of length $12$ and minimum distance at least~$5$ contains at most $1838$ codewords, which is a contradiction.
\par
Now assume that $C_O$ is not $S$-optimal. Then there exists a $[29,9,15]_3$ code $C$. Its dual~$C^\perp$ is a $[29,20,d^\perp]_3$ code with minimum distance $d^\perp \le  6$ by the sphere packing bound. Let $c\in C^\perp$ be a codeword of minimum nonzero weight. Shortening $C$ in a set of~$6$ positions containing the support of $c$ gives a $[23, \ge 4, \ge 15]_3$ code. After puncturing in two positions we get a $[21,\ge 4, \ge 13]_3$ code, which implies the existence of a $[21,4,13]_3$-code. This code does not exist by~\cite[Thm.~3.2]{HillNewton1992OptimalTernaryLinearCodes}.
\end{proof}
\par
We remark that for $q = 5$, the code $C_O$ might also be $P$-optimal (and hence strongly optimal), since according to \cite{codetables} the existence of a linear $[127,8,45]_5$ code is open.
\par
Based on the above investigation, we dare to state the following conjecture.
\begin{conjecture}
The code $C_O$ is strongly optimal for all $q\ge 3$.
\end{conjecture}
\par
We close this section with some remarks on codes obtained by puncturing $C_O$. Let $T$ be a set of $t$ coordinate positions of $C_O$. Deleting all coordinates in $T$ in each codeword of~$C_O$, we obtain a linear code $C_O^T$ of length $q^3+1-t$. If $t$ is strictly less than $q^3-q^2-q$ (the minimum distance of $C_O$), then the dimension of $C_O^T$ is still $8$ and the minimum distance of $C_O^T$ is at least $q^3-q^2-q-t$.
\par
One can show that, for $t\le \lfloor (q-3)/2\rfloor$ and $q\ge 5$, the code $C_O^T$ has minimum distance exactly $q^3-q^2-q-t$ and is still $n$-optimal. To prove this, it is enough to take $t=\lfloor (q-3)/2\rfloor$ and to use the same approach as that in the proof of Theorem~\ref{thm:lp_bound} with
\[
z_1 = q^3 - q^2 - q - t,\quad
z_2 = q^3 - q^2 + q - 2 - t,\quad
z_3 = q^3 - q^2 + q - 1 - t, \quad\text{and}\quad
n = q^3 - t
\]
to show that no linear code with parameters $[q^3-t,8,\ge q^3-q^2-q-t]_q$ can exist.

\section{The dual code $C_O^\perp$}

We now investigate the dual code $C_O^\perp$ of $C_O$ of length $q^3+1$ and dimension $q^3-7$.
\begin{thm}
The code $C_O^\perp$ has parameters $[q^3+1,q^3-7,d]_q$, where
\[
d = \begin{cases}
9 & \text{for $q = 2$};\\
6 & \text{for $q = 3$};\\
5 & \text{for $q\ge 4$}.
\end{cases}
\]
Moreover the number of codewords in $C_O^\perp$ of weight $5$ is
\[
\frac{1}{120}(q-3)(q-2)(q-1) q^3 (q^6-1).
\]
\end{thm}
\begin{proof}
Let $(A_i)_i$ and $(A'_i)_i$ be the weight distributions of $C_O$ and $C_O^\perp$, respectively. By the MacWilliams identity we have
\[
A'_j = \frac{1}{|C|} \sum_{i=0}^n K_j(i) A_i\quad\text{for each $j$}.
\]
Now $(A_i)_i$ is given in Theorem~\ref{thm:CO} and a computation reveals that $A'_1=A'_2=A'_3=A'_4=0$ and that $A'_5$ equals the corresponding expression in the statement of the theorem. This shows that $A'_5\ne 0$ for $q\ge 4$. For $q=3$ we have $A'_5=0$ and we compute $A'_6=6552$. For $q=2$, the code $C_O$ is the binary parity check code of length $9$ whose dual is the binary repetition code of length $9$ and minimum distance $9$.
\end{proof}
\par
\begin{thm}
The code $C_O^\perp$ is $n$-optimal for all $q$.
\end{thm}
\begin{proof}
The claim is clear for $q = 2$. For $q = 3$, shortening a $[27,20,6]_3$ code in $12$ positions would result in a $[17,\ge 8, \ge 6]_3$ code, which does not exist by~\cite{VEupen1995}.
\par
Now let $q\ge 4$ and suppose that $C$ is a code of length $q^3$ and minimum distance $5$ over $\mathbb{F}_q$. A Hamming sphere of radius $2$ in $\mathbb{F}_q^{q^3}$ is of size
\[
\sum_{i=0}^2 \binom{q^3}{i} (q-1)^i = \tfrac{1}{2}q^8 - q^7 + \tfrac{1}{2}q^6 - \tfrac{1}{2}q^5 + 2q^4 - \tfrac{3}{2}q^3 + 1.
\]
Since $q \ge 4$, one checks that this is strictly larger than $q^7$. Hence by the sphere packing bound we have $|C|<q^{q^3 - 7}$. Hence no $[q^3, q^3-7,5]_q$ code can exist. This implies $n$-optimality of $C_O^\perp$.
\end{proof}
\par
For $q = 2$, the code $C_O^\perp$ is $S$-optimal, but not $P$-optimal. Moreover from the tables~\cite{codetables} we find the following.  For $q = 3$, the code $C_O^\perp$ is $P$-optimal, but $S$-optimality is open. For $q = 4$, the code $C_O^\perp$ is not $S$-optimal and $P$-optimality is open. For $q = 5$, both $S$- and $P$-optimality is open.

\end{document}